\documentclass{amsproc}
\usepackage{amsmath}
\usepackage{cite}
\usepackage{footnote}
\usepackage{amssymb}
\usepackage{mathrsfs}
\usepackage[dvipspdf]{xcolor}

\def\zhou#1 {\fbox {\footnote {\ }}\ \footnotetext { From Zhou: {\color{red}#1}}}

\def\alex#1 {\fbox {\footnote {\ }}\ \footnotetext { From Alex: {\color{blue}#1}}}

\def\qi#1 {\fbox {\footnote {\ }}\ \footnotetext { From Qi: {\color{orange}#1}}}

\newcommand{\bZ}{{\bf Z}} 
\newcommand{\tr}{{\rm Tr}} 

\newcommand{\gf}{{\mathbb F}}

\newcommand{\calp}{{\mathcal P}}
\newcommand{\calb}{{\mathcal B}}
\newcommand{\cals}{{\mathcal S}}

\newcommand{\calf}{{\mathcal F}}
\newcommand{\calc}{{\mathcal C}}
\newcommand{\lc}{{\rm LC}}

\newtheorem{theorem}{Theorem}[section]
\newtheorem{corollary}{Corollary}

\newtheorem{lemma}[theorem]{Lemma}
\newtheorem{proposition}{Proposition}

\theoremstyle{definition}

\newtheorem{example}{Example}
\newtheorem{remark}{Remark}

\title[Sets of ZDB Functions and Applications]
      {Sets of Zero-Difference Balanced Functions and Their Applications}

\author[Qi Wang and Yue Zhou]{}


\email{qi.wang@ovgu.de}
\email{yue.zhou.ovgu@gmail.com}

\thanks{The first author's research was supported by the Alexander von Humboldt (AvH) Stiftung/Foundation. The second author was affiliated with Institute of Algebra and Geometry, Otto-von-Guericke University Magdeburg, 39106 Magdeburg, Germany.}

\begin{document}

\maketitle



\centerline{\scshape Qi Wang}
\medskip
{\footnotesize
 \centerline{Institute of Algebra and Geometry}
   \centerline{Otto-von-Guericke University Magdeburg}
   \centerline{39106 Magdeburg, Germany}
}

\bigskip
\centerline{\scshape Yue Zhou}
\medskip
{\footnotesize
 \centerline{Department of Mathematics and System Sciences}
   \centerline{National University of Defense Technology}
   \centerline{410073, Changsha, China}
}

\bigskip

 \centerline{(Communicated by the associate editor name)}

\begin{abstract}
Zero-difference balanced (ZDB) functions can be employed in many applications, e.g., optimal constant composition codes, optimal and perfect difference systems of sets, optimal frequency hopping sequences, etc. In this paper, two results are summarized to characterize ZDB functions, among which a lower bound is used to achieve optimality in applications and determine the size of preimage sets of ZDB functions. As the main contribution, a generic construction of ZDB functions is presented, and many new classes of ZDB functions can be generated. This construction is then extended to construct a set of ZDB functions, in which any two ZDB functions are related uniformly. Furthermore, some applications of such sets of ZDB functions are also introduced.
\end{abstract}

\section{Introduction}\label{sec-intro}

Let $(A, +)$ and $(B,+)$ be two abelian groups of orders $n$ and $\ell$, respectively. For a function $f$ from $A$ onto $B$, define
$$
N_b(a) := \big| \{x \in A: f(x+a) - f(x) = b \} \big|.
$$
If $N_b(a) = \frac{n}{\ell}$ for all $b \in B$ and all nonzero $a \in A$, the function $f$ is called {\em planar} or {\em perfect nonlinear}~\cite{DO68,Ny91}. If $N_0(a) = \frac{n+1}{\ell} - 1$ for each nonzero $a \in A$ and $N_b(a) = \frac{n+1}{\ell}$ for each nonzero $b \in B$ and each nonzero $a \in A$, $f$ is called a {\em difference balanced} function~\cite{GG05,ZTWY12}. Here we consider a relaxation of these two types of functions: if $N_0(a) = \lambda$ for all nonzero $a \in A$, where $\lambda$ is a nonnegative integer, the function $f$ is called an $(n, \ell, \lambda)$-{\em zero-difference balanced} (ZDB) function.  

Zero-difference balanced (ZDB) functions were first defined by Ding~\cite{Ding09}, and since then have found many applications: they can be used to construct optimal and perfect difference systems of sets~\cite{Ding09,ZTWY12}, optimal constant composition codes~\cite{DY051,DY05,Ding08}, etc. For the background of difference systems of sets, we refer to~\cite{Ding09,Lev71,Lev04,Wang06}, and for more information on constant composition codes, see~\cite{Ding08,DY05,Luo03}. In design theory, ZDB functions correspond to partitioned difference families.

Let $(A, +)$ be an abelian group of order $n$. Let $\calp$ be a collection of $\ell$ subsets ({\em blocks}) $\calb_0, \calb_1, \ldots, \calb_{\ell-1}$ of $A$. The collection $\calp$ is said to be an $(n, K, \lambda)$-{\em difference family} (DF) in $A$, where $K = \{*\ |\calb_i|: 0 \le i < \ell \ *\}$, if for $0 \leq i < \ell$, the list of differences $b - b'$, with $b, b' \in \calb_i$ and $b \ne b'$, covers all nonzero elements in $A$ exactly $\lambda$ times. Furthermore, if $\calp$ forms a partition of $A$, it is called an $(n, K, \lambda)$-{\em partitioned difference family} (PDF). Clearly, ZDB functions and PDFs are basically two equivalent objects.

\begin{proposition}\label{pro-zdb-pdf}
  Let $(A, +)$ and $(B, +)$ be two abelian groups of orders $n$ and $\ell$, respectively, where $B = \{ b_0, b_1, \ldots, b_{\ell-1} \}$. Let $f$ be a function from $A$ onto $B$. Define $\calb_i := \{ x \in A: f(x) = b_i\}$ for $0 \leq i < \ell $, and $\calp = \{ \calb_0, \calb_1, \ldots, \calb_{\ell-1} \}$. Then $f$ is an $(n, \ell, \lambda)$-ZDB function if and only if $\calp$ is an $(n, K, \lambda)$-PDF, where $K = \{*\ | \calb_i |: 0 \leq i < \ell \ * \}$.
\end{proposition}

Recently, Zhou, Tang, Wu and Yang~\cite{ZTWY12} constructed some new classes of ZDB functions from difference balanced functions, and then presented several applications. For more information on ZDB functions, we also refer to a recent survey~\cite{DT12}. In this paper, we are mainly concerned with new classes of single ZDB functions, new sets of ZDB functions, and applications of sets of ZDB functions. The remainder of the present paper is organized as follows. In
Section~\ref{sec-char}, we present two results to characterize ZDB functions. We then propose a generic construction of ZDB functions in Section~\ref{sec-const}, which can give many new classes of ZDB functions. In Section~\ref{sec-exd}, we extend this generic construction naturally to construct a set of ZDB functions, in which any two ZDB functions are related uniformly. In Section~\ref{sec-app}, we give two applications of such sets of ZDB functions. We then conclude this paper with some open problems in Section~\ref{sec-con}.

Throughout this paper, if not stated otherwise, we use the following notations:
\begin{enumerate}
  \item[--] $q$ is a prime power.
  \item[--] $m$ is a positive integer.
  \item[--] $\theta$ is a primitive element of $\gf_{q^m}$.
  \item[--] $\bZ_n = \{0, 1, 2, \ldots, n -1 \}$ associated with the integer addition modulo $n$ and integer multiplication modulo $n$ operations.
  \item[--] $\tr$ denotes the trace function from $\gf_{q^m}$ to $\gf_q$.
  \item[--] $\lceil x \rceil$ denotes the ceiling function, and $\lfloor x \rfloor$ is the floor function.
\end{enumerate}


\section{Characterizations of ZDB functions}\label{sec-char}

In this section, to characterize ZDB functions, we give two results: a lower bound on the parameter $\lambda$ of ZDB functions, and general bounds on the size of preimage sets of ZDB functions.

\subsection{A lower bound on $\lambda$}\label{subsec-lowerbound}

Let $(A, +)$ and $(B,+)$ be two abelian groups of orders $n$ and $\ell$, respectively, where $B = \{b_0, b_1, \ldots, b_{\ell - 1}\}$. Suppose that $f$ is an $(n, \ell, \lambda)$-ZDB function from $A$ onto $B$. To characterize ZDB functions, we have the following result directly from the definition of PDF and Proposition~\ref{pro-zdb-pdf}.

\begin{lemma}\label{lem-zdbcha}
  Define 
  $\calb_i := \{ x \in A: f(x) = b_i\}$ for $0 \leq i < \ell$. Then 
  \begin{equation*}
    \left\{ \begin{array}{l}
      \sum_{i=0}^{\ell - 1} \tau_i = n, \\
      \sum_{i=0}^{\ell - 1} \tau_i^2 = n + \lambda (n-1),
    \end{array} \right.
  \end{equation*}
  where $\tau_i = |\calb_i|$ for $0 \leq i < \ell$.
\end{lemma}

Based on the two equations above, we have the following lower bound on $\lambda$. 

\begin{lemma}\label{lem-lowbound}
For any $(n, \ell, \lambda)$-ZDB function $f$ from $A$ onto $B$, we have
\begin{equation}\label{eqn-lowbound}
  \lambda \geq \left\lceil \frac{(n - \epsilon) (n + \epsilon - \ell)}{\ell (n - 1)} \right\rceil , 
\end{equation}
where $n = k \ell + \epsilon$ with $0 \leq \epsilon < \ell$. In particular, 
$$
  \lambda = \frac{(n - \epsilon) (n + \epsilon - \ell)}{\ell (n - 1)} 
$$
if and only if, for $0 \leq i < \ell $, $\tau_i = k$ for $\ell - \epsilon$ times and $\tau_i = k+1$ for the other $\epsilon$ times.
\end{lemma}


\begin{proof}
  By Lemma~\ref{lem-zdbcha}, we have
  \begin{equation*}
    \lambda \geq \frac{1}{n-1} \left( \min \sum_{i=0}^{\ell - 1} \tau_i^2 - n\right).
  \end{equation*}
 
Note that $\sum_{i=0}^{\ell - 1} \tau_i = n$. By integral programming, $\{\tau_0, \tau_1, \ldots, \tau_{\ell - 1} \}$ attains the minimum value if and only if $f$ is as balanced as possible. Since $n = k \ell + \epsilon$, if and only if $\tau_i = k$ for $\ell - \epsilon$ times and $\tau_i = k +1$ for the other $\epsilon$ times, we obtain the lower bound of $\lambda$ as stated. 
 \end{proof}


\begin{remark}\label{rmk-lowbound}
  Since the bound of (\ref{eqn-lowbound}) coincides with the bound on frequency hopping sequences in~\cite[Lemma 4]{LG74} (see also Lemma~\ref{lem-fhsbound}), ZDB functions meeting the lower bound of (\ref{eqn-lowbound}) can be used to define optimal frequency hopping sequences (e.g., see~\cite{DMY07,DY08,FMM04,GFM06,GMY09}). Furthermore, by~\cite[Proposition 3]{Ding08} and~\cite[Lemma 6]{ZTWY12}, if there exists an $(n, \ell, \lambda)$-ZDB function achieving the bound of (\ref{eqn-lowbound}), the corresponding constant composition codes and difference systems of sets are both optimal.
\end{remark}




\subsection{General bounds on the size of preimage sets}\label{subsec-pisbound}

Using Lemma~\ref{lem-lowbound}, we can explicitly determine the size of preimage sets of an $(n, \ell, \lambda)$-ZDB function for a specific $\lambda$ prescribed as in Lemma~\ref{lem-lowbound}. Now we give general bounds on the size of preimage sets of ZDB functions. The sizes of all preimage sets constitute the parameter $K$ in the corresponding PDF, and are also important in applications.

\begin{lemma}\label{lem-pisbound}
Suppose that $f$ is an $(n, \ell, \lambda)$-ZDB function from $(A,+)$ onto $(B,+)$. For each $0 \leq i < \ell$, we have
\begin{equation}\label{eqn-pisbound}
  \frac{n - \sqrt{\Delta}}{\ell} \leq \tau_i \leq \frac{n + \sqrt{\Delta}}{\ell}  , 
\end{equation}
where $\Delta = (n + \lambda n - \lambda) \ell^2 - ( n^2 + n + \lambda n - \lambda) \ell + n^2$. In particular, 
\begin{itemize}
  \item if $\displaystyle\lambda = \frac{n}{\ell}$, we have $\displaystyle\frac{n - (\ell - 1)\sqrt{n}}{\ell} \leq \tau_i \leq \frac{n + (\ell - 1)\sqrt{n}}{\ell}$ ; 
  \item if $\displaystyle\lambda = \frac{n+1}{\ell} - 1$, we have
    $\displaystyle\frac{n - \ell + 1}{\ell} \leq \tau_i \leq \frac{n + \ell - 1}{\ell}$.
\end{itemize}
\end{lemma}

\begin{proof}
Without loss of generality, it suffices to prove the bound for $\tau_0$. Note that
\begin{eqnarray*}
  0 & \leq & \sum_{\begin{subarray}{c} 1 \leq i, j < \ell \\ i \ne j \end{subarray}} (\tau_i - \tau_j)^2 \\
    & = & \sum_{\begin{subarray}{c} 1 \leq i, j < \ell \\ i \ne j \end{subarray}} (\tau_i^2 + \tau_j^2 - 2\tau_i \tau_j) \\
      & = & 2(\ell - 2) \sum_{i=1}^{\ell-1} \tau_i^2 - 2 \sum_{\begin{subarray}{c} 1 \leq i, j < \ell \\ i \ne j  \end{subarray}} \tau_i \tau_j .
      \end{eqnarray*}
It then follows that
\begin{equation}\label{eqn-genbound1}
  (\ell - 2) \sum_{i=1}^{\ell - 1} \tau_i^2 \geq \sum_{\begin{subarray}{c} 1 \leq i, j < \ell \\ i \ne j \end{subarray}} \tau_i \tau_j .
  \end{equation}
  By Lemma~\ref{lem-zdbcha}, we have
  \begin{eqnarray}
    \lefteqn{n + \lambda (n - 1) } \nonumber \\
    & = & \sum_{i=0}^{\ell  - 1} \tau_i^2 - \tau_0^2 + \tau_0^2 \nonumber \\
    & = & \sum_{i=1}^{\ell - 1} \tau_i^2 + \left( n - \sum_{i=0}^{\ell - 1} \tau_i + \tau_0 \right)^2 \nonumber \\
    & = & \sum_{i=1}^{\ell - 1} \tau_i^2 + \left(n - \sum_{i=1}^{\ell - 1} \tau_i \right)^2 \nonumber \\
    & = & 2 \sum_{i=1}^{\ell - 1} \tau_i^2 + n^2 - 2n \sum_{i=1}^{\ell - 1} \tau_i + \sum_{\begin{subarray}{c} 1 \leq i, j < \ell \\ i \ne j \end{subarray}} \tau_i \tau_j . \label{eqn-genbound2}
    \end{eqnarray}
    With (\ref{eqn-genbound1}) and (\ref{eqn-genbound2}), we have
    $$
    \ell \sum_{i=1}^{\ell - 1} \tau_i^2 - 2n \sum_{i=1}^{\ell - 1} \tau_i + n^2 \geq n + \lambda (n - 1).
    $$
    Applying Lemma~\ref{lem-zdbcha}, we obtain
    $$
    \ell (n + \lambda (n-1) - \tau_0^2) - 2n (n - \tau_0) + n^2 \ge n +  \lambda(n-1).
    $$
    It then follows that
    $$
    (\tau_0 - \frac{n}{\ell})^2 \leq \frac{\Delta}{\ell^2},
    $$
    where $\Delta = (n + \lambda n - \lambda) \ell^2 - (n^2 + n + \lambda n - \lambda) \ell + n^2$, which completes the proof.
\end{proof}

\begin{remark}\label{rmk-pisbound}
  The two special cases in Lemma~\ref{lem-pisbound} correspond to perfect nonlinear functions and difference balanced functions, respectively. For the case of perfect nonlinear functions, the bounds were also given in~\cite{CDY05}.  
\end{remark}

\section{A generic construction of ZDB functions}\label{sec-const}

In this section, we describe a generic construction of ZDB functions, and present two special cases of this construction. 


\subsection{The construction}

To present the construction of ZDB functions, we need the following results.

\begin{lemma}\label{lem-pre}
Let $e = l \cdot r$ be a divisor of $q - 1$ with $\gcd(e,m) = 1$. Define $D_0 := \langle \theta^r \rangle$, $C_0 := \langle \theta^e \rangle$ and $\alpha = \theta^{\frac{q^m-1}{q-1}}$. Then 
 $$
 \gf_{q^m}^* = \dot\bigcup_{i=0}^{r-1} D_i ,
  $$
  and
  $$
  D_0 = \dot\bigcup_{i=0}^{l-1} C_i ,
  $$
  where $ D_i = \alpha^i D_0 $ for $0 \leq i < r$, $ C_i =  \alpha^{ir} C_0$ for $0 \leq i < l$, and $\dot\bigcup$ denotes the disjoint union. 
\end{lemma}

\begin{proof}
  Since the first assertion is a special case of the second one, we only need to prove the second assertion. Note that $\alpha = \theta^{\frac{q^m-1}{q-1}}$ is a primitive element of $\gf_q$. Since $|D_0| = l \cdot |C_0|$, it suffices to prove that $\alpha^{ir} \not \in C_0$ for all $i = 1, \ldots, l-1$. Assume to the contrary that there exists some $j$ such that $\alpha^{jr} \in C_0$, we then have $\alpha^{j r \cdot \frac{q^m - 1}{e}} = 1$, which means 
  $$
  jr \cdot \frac{q^m-1}{e} \equiv 0 \pmod{(q-1)}.
  $$
  It follows that
  $$
  jr \cdot \frac{q^m - 1}{q-1} \equiv 0 \pmod{e}.
  $$
  Since $e$ is a divisor of $q-1$, we have $q \equiv 1 \pmod{e}$. Thus,
  $$
  jr \cdot \frac{q^m - 1}{q-1} \equiv jr \cdot m \pmod{e}.
  $$
  We then obtain that $jr \cdot m  \equiv 0 \pmod{e}$, which implies that $e | jr$ since $\gcd(e,m) = 1$. This is a contradiction to the choice of $j$, i.e., $0 < j \le l -1$. Therefore, $\alpha^{ir} C_0$ for $i = 0, 1, \ldots, l-1$ are pairwise disjoint. The proof is then completed.  
\end{proof}

\begin{corollary}\label{coro-pre}
  With the same notations as in Lemma~\ref{lem-pre}, assume that $h$ is a $d$-homogeneous function on $\gf_{q^m}^*$ over $\gf_q$, i.e., for all $a \in \gf_q$ and $x \in \gf_{q^m}^*$, $h(ax) = a^d h(x)$. Then we have 
  $$
  \big|\{x\in D_0: h(x)=0\}\big| = l \cdot \big| \{x\in C_i: h(x)=0\} \big|, 
  $$
for each $i = 0, 1, \ldots, l -1$.
\end{corollary}

\begin{proof}
  Let $x_0\in C_0$ be a root of $h(x)=0$, then for each $0 \leq i < l$, $\alpha^{ir} x_0\in C_i$ is also a root of it, because
  $$
  h(\alpha^{ir} x_0)=\alpha^{ird}h(x_0)=0.
  $$
    
  Since by Lemma~\ref{lem-pre} $D_0 = \dot\bigcup_{i=0}^{l-1}C_i = \dot\bigcup_{i=0}^{l-1} \alpha^{ir} C_0$, all the solutions of $h(x)=0$ in $D_0$ are equally distributed into each of the $l$ cosets $C_i$'s. Thus, we have
    $$
    \big| \{x\in D_0: h(x)=0\}| = l \cdot |\{x\in C_i: h(x)=0\} \big|
    $$
    for each $i=0, 1, \ldots, l-1$.
\end{proof}

\begin{lemma}\label{lem-pre2}
  With the same notations as in Lemma~\ref{lem-pre}, let $u$ be a divisor of $q-1$ with $\gcd(u,m) = 1$. Define
    $$
    N_{a,i} := \big| \{ x \in C_i: \tr(ax^u) = 0 \} \big| ,
    $$
    then for each $a \in \gf_{q^m}^*$ and $ 0 \leq i < l$, we have
    $$
    N_{a,i} = \frac{q^{m-1} - 1}{l \cdot r}.
    $$
\end{lemma}

\begin{proof}
  Since $\tr(ax)$ is a $1$-homogeneous function on $\gf_{q^m}^*$ over $\gf_q$ for each $a \in \gf_{q^m}^*$, by Corollary \ref{coro-pre},  we have 
    $$
    \big|\{x \in \langle \theta^u \rangle: \tr(a x)=0\} \big|=\frac{q^{m-1}-1}{u},
    $$
    which implies that
    $$
    \big| \{ 0 \leq j < \frac{q^m-1}{u} : \tr(a \theta^{uj}) = 0\} \big| = \frac{q^{m-1}-1}{u},  
    $$
    and further
    $$
    \big| \{x \in\gf_{q^m}^*: \tr(a x^u)=0\}\big| = q^{m-1}-1.
    $$

    Since $\tr(ax^u)$ is a $u$-homogeneous function on $\gf_{q^m}^*$ over $\gf_q$ for each $a \in \gf_{q^m}^*$, applying Corollary \ref{coro-pre} again, we have
    $$
    \big| \{x\in D_0: \tr(a x^u)=0\} \big| = \frac{q^{m-1}-1}{r}.
    $$
    Thus, 
  \begin{equation}\label{eqn-con1}
    N_{a ,i}:= \big|\{x\in C_i: \tr(a x^u)=0\} \big|=\frac{q^{m-1}-1}{l \cdot r},
  \end{equation}
  for each $a \in \gf_{q^m}^*$ and $0 \leq i < l$, which completes the proof.
\end{proof}

Now we are ready to present a generic construction of ZDB functions with parameters $\left( \frac{q^m - 1}{r}, q, \frac{q^{m-1}-1}{r} \right)$, where $r$ is a divisor of $q-1$ with $\gcd(r,m) = 1$.

\begin{theorem}\label{thm-const1}
  Let $e$ and $u$ be two divisors of $q-1$ with $\gcd(e, m)=\gcd(u,m) = 1$ and $e = l\cdot r$. Set $D_0 = \langle \theta^r \rangle$, $C_0 = \langle \theta^e \rangle$, and $\alpha = \theta^{\frac{q^m-1}{q-1}}$. Define the function $f: (\bZ_n, +) \rightarrow (\gf_q,+)$ by
  $$
  f(t) := \tr(\rho(t) \theta^{rut}),
  $$
  where $n = \frac{q^m-1}{r}$ and $\rho(t)$ is defined as
  $$
  \rho(t) := d_i, \ \textrm{ if $\theta^{rt} \in C_i$,}
  $$
  with $C_i = \alpha^{i r} C_0$ and $d_i \in \gf_{q^m}^*$ for $0 \leq i < l$ . If the following two conditions
  \begin{enumerate}
    \item[(i)] $\{x \in C_0 : x^u=1 \textrm{ and } x\neq 1\} = \emptyset$;
    \item[(ii)] $d_j / d_{k+j} \not \in C_{uk}$ for each $k \ne 0$ and $0 \le j < l$, where the subscripts $uk$ and $k+j$ are performed modulo $l$,
  \end{enumerate}
are satisfied, the function $f(t)$ is a $\left(\frac{q^m - 1}{r}, q, \frac{q^{m-1} - 1}{r}\right)$-ZDB function.
\end{theorem}

\begin{proof}
  By definition, we need to prove 
  $$
  N_0(a) = \big| \{ t \in \bZ_n: f(t+a) - f(t) = 0 \} \big| = \frac{q^{m-1}-1}{r}
  $$
  for each nonzero $a \in \bZ_n$. To this end, without loss of generality, assume that $\theta^{ r a } \in C_k$ for some $0 \le k < l$. By Lemma~\ref{lem-pre}, we then have
  \begin{eqnarray*}
    \lefteqn{ \big| \{ t \in \bZ_n : f(t + a) - f(t) = 0 \} \big|} \\
    & = & \big| \{ t \in \bZ_n: \tr \left( ( \rho(t+a) \theta^{ra u} - \rho(t) ) \theta^{rut} \right) = 0 \} \big| \\
    & = & \sum_{j=0}^{l-1} \big| \{ x \in C_j: \tr \left( ( d_{k+j} \theta^{rau} - d_j ) x^u \right) = 0 \} \big| .
  \end{eqnarray*}

  On one hand, if $k = 0$, i.e., $\theta^{ra} \in C_0$, since $\{ x \in C_0: x^u=1 \textrm{ and } x\neq 1 \} = \emptyset$, we have $d_j \theta^{rau} - d_j \ne 0$ for each nonzero $a \in \bZ_n$ and each $0 \le j < l$. On the other hand, if $k \ne 0$, we have $\theta^{rau} \in C_{uk}$, where $uk \not \equiv 0 \bmod{l}$. Since $d_j / d_{k+j} \not \in C_{uk}$ for $0 \leq j < l$, we also have $d_{k+j} \theta^{rau} - d_j \ne 0$ for each nonzero $a \in \bZ_n$ and each $0 \le j < l$. Thus, from Lemma~\ref{lem-pre2}, it follows that   
  \begin{eqnarray*}
   \lefteqn{ \big| \{ t \in \bZ_n : f(t + a) - f(t) = 0 \} \big|} \\
   & = & \sum_{j = 0}^{l-1} N_{d_{k+j} \theta^{rau} - d_j, j} \\ 
   & = & \frac{q^{m-1} - 1}{r} .
  \end{eqnarray*}
  The proof is then completed.
\end{proof}

In Theorem~\ref{thm-const1}, we presented the ZDB function $f$ from $(\bZ_n, +)$ onto $(\gf_q, +)$. Since $D_0 \cong (\bZ_n,+)$ where $n = \frac{q^m-1}{r}$, in the sequel sometimes we use the multiplicative group $D_0$ instead of $(\bZ_n, +)$. We hope that this would not bring any confusion.

\begin{remark}\label{rmk-const1}
  The two sufficient conditions in Theorem~\ref{thm-const1} can be satisfied. 
      It is easily checked that the condition (i) is equivalent to that for all $1 \leq j < \frac{q^m-1}{e}$, the relation $j \cdot e \cdot u \not\equiv 0 \pmod{q^m-1}$ holds, of which $u = 1$ is a simple example. Thus, the condition (i) always holds by choosing suitable $e$, $u$ and $r$. By Lemma~\ref{lem-pre}, we have
      $$
      \gf_{q^m}^* = \dot\bigcup_{i=0}^{r-1} \alpha^i D_0 = \dot\bigcup_{i=0}^{lr - 1} \alpha^i C_0, 
      $$
      where $\alpha = \theta^{\frac{q^m-1}{q-1}}$. 
      If $d_j \in \alpha^{j_1}D_0$ and $d_{k+j} \in \alpha^{j_2}D_0$ with $0 \leq j_1 \ne j_2 \leq r-1$, the condition (ii) is always satisfied. We now consider two extreme cases: 
      \begin{itemize}
	\item suppose that $d_i \in D_0$ for each $0 \leq i < l$, i.e., $d_i \in \alpha^{-s_i r}C_0$ with $0 \leq s_i < l$. Then the condition (ii) is equivalent to
      $$
      -s_j+s_{k+j}\not \equiv uk \pmod{l},
      $$
      for all $k\neq 0$ and $0\le j < l$, which can be also written as $s_j-s_i\not \equiv u(j-i) \pmod{l}$, i.e.,
      $$(s_j-ju)-(s_i-iu)\not \equiv 0 \pmod{l},$$
      for all $j\neq i$ and $0\le i,j < l$. Hence the condition (ii) can be expressed as
      $$
      \{s_i-iu \pmod{l}: 0\le i < l \}=\{0,1,\cdots, l-1\},
      $$
      and there are totally $ l! |C_0|^l$ different $\rho(t)$'s satisfying this condition. 
    \item suppose that $l \ge r$. Let each of $r-1$ different $d_i$'s belong to each of $r-1$ different cyclotomic classes $D_i$'s. There are ${l \choose {r-1}}$ ways to do this. If $d_j, d_{k+j}$ don't belong to the same $D_i$, the condition (ii) is always satisfied. Thus, for these $r-1$ $d_i$'s, there are ${l \choose {r-1}} |D_0|^{r-1}$ possible choices. Now we only need to consider the remaining $l-r+1$ $d_i$'s, which belong to the rest one cyclotomic class $D_0$ without loss of generality. With similar argument, there are totally $ { l \choose {r - 1}} (l-r+1)! |C_0|^{l-r+1} |D_0|^{r-1} $ different $\rho(t)$'s. 
  \end{itemize}
    Thus, there are always exponentially many $\rho(x)$'s satisfying the condition (ii).
\end{remark}

\subsection{Two special cases}

By Remark~\ref{rmk-const1}, the construction in Theorem~\ref{thm-const1} is generic in the sense that we can choose different $\rho(x)$, $u$, $e$ and $r$ to get many new classes of ZDB functions. Now we give two special cases of the construction in Theorem~\ref{thm-const1}, which in fact extended the previously known constructions~\cite{ZTWY12,Ding09,Ding08}.

\subsubsection{Special case I}

Let $q$ be an odd prime power, $m$ be odd, $e=2$, and $u = r = 1$. We have the following construction of ZDB functions. 
 
\begin{corollary}\label{coro-sc1}
  Let $q$ be an odd prime power and $m$ be an odd integer. Define the function $f:\ \gf_{q^m}^* \rightarrow \gf_q$ as
$$
f(x) := \tr( \rho(x) x ),
$$
where $\rho(x)$ is defined as 
\begin{equation*}
  \rho(x) := \left\{ \begin{array}{ll}
    d_0, & \textrm{ if $x$ is a square in $\gf_{q^m}^*$,} \\
    d_1, & \textrm{ if $x$ is a nonsquare in $\gf_{q^m}^*$},
  \end{array} \right.
\end{equation*}
with $d_0, d_1 \in \gf_{q^m}^*$. If $d_0 d_1$ is a square, then the function $f$ is a $(q^m-1, q, q^{m-1}-1)$-ZDB function. Furthermore, if $q^m$ is large enough, when $d_0 \ne \pm d_1$, we can always choose suitable $d_0$ and $d_1$ such that for each square $\delta \in \gf_{q^m} \setminus \{0, 1\}$, $N_b(\delta) = q^{m-1}$, and for some nonsquare $\delta \in \gf_{q^m} \setminus \{0, 1\}$, $N_b(\delta) \ne q^{m-1}$ for all $b \in \gf_q^*$, i.e., the function $f(x)$ is not difference
balanced, where 
$$
N_b(\delta) := \big| \{ x \in \gf_{q^m}^* : f(\delta x ) - f(x) = b \} \big| .
$$
\end{corollary}

The first argument of Corollary~\ref{coro-sc1} directly follows from Theorem~\ref{thm-const1}. To prove the second one, we need some results on quadratic forms over $\gf_q$. A {\em quadratic form} in $m$ indeterminates over $\gf_q$ is a homogeneous polynomial in $\gf_q[x_1, \ldots, x_m]$ of degree $2$ or the zero polynomial. If $q$ is odd, any quadratic form $f$ over $\gf_q$ can be represented as 
$$
f(x_1, \ldots, x_m) = \sum_{i,j = 1}^m a_{ij} x_i x_j, \textrm{ with $a_{ij} = a_{ji}$}.
$$
The matrix $A = (a_{ij})_{m\times m}$ associated with $f$ is called the {\em coefficient matrix} of $f$.

\begin{lemma}~\cite[Theorem 6.27]{LN97}\label{lem-quaform}
Let $f$ be a non-degenerate quadratic form over $\gf_q$, $q$ odd, in an odd number $m$ of indeterminates. Then for $b \in \gf_q$, the number of solutions of the equation $f(x_1, \ldots, x_m) = b$ in $\gf_q^m$ is
$$
q^{m-1} + q^{(m-1)/2} \eta\left( (-1)^{(m-1)/2} b \Delta \right),
$$
where $\eta$ is the quadratic character of $\gf_q$,  $\Delta = \det(A)$ and $A$ is the coefficient matrix of $f$. 
\end{lemma}

\begin{lemma}~\cite{CW66}\cite[Exercise 6.72]{LN97}\label{lem-numsol}
  Let $a_1, a_2, b_1, b_2 \in \gf_q^*$ with $a_1 b_2 \ne a_2 b_1$ where $q$ is a prime power and let $n, n_1, n_2 \in \mathbb{N}$. The number $N$ of common solutions $(x_1, x_2, x_3) \in \gf_q^3$ of the equations
  \begin{equation*}
    \left\{ \begin{array}{l}
      x_1^{n_1} = a_1 + b_1 x_3^n \\
      x_2^{n_2} = a_2 + b_2 x_3^n 
    \end{array} \right.
  \end{equation*}
  satisfies $| N - q| \le C q^{1/2}$ for some constant $C$ independent of $q$.
\end{lemma}

\begin{lemma}\label{lem-pre3}
  Let $q$ be an odd prime power and $m$ be an odd integer. For each $\delta \in \gf_{q^m}^*$, the equation $\tr(\delta x^2) = 0 $ has exactly $q^{m-1}$ solutions in $\gf_{q^m}$, and the equation $\tr(\delta x^2) = b$, with $b \in \gf_q^* $, has exactly $q^{m-1} \pm q^{(m-1)/2}$ solutions depending on the quadratic characters of $\delta$ and $b$. Furthermore, if the equation $\tr(\delta x^2) = b$, for some $\delta \in \gf_{q^m}^*$ and $b \in
  \gf_q^*$, has exactly $q^{m-1} + q^{(m-1)/2} $ solutions,  then the equation $\tr(a \delta x^2) = b$ has exactly $q^{m-1} - q^{(m-1)/2}$ solutions, where $a \in \gf_q^*$ is a nonsquare, and vice versa.
\end{lemma}

\begin{proof}
Note that the bilinear form
\begin{equation*}
  B(x,y) = \tr(\delta(x+y)^2) - \tr(\delta x^2) - \tr( \delta y^2) = \tr(2\delta xy)
\end{equation*}
is non-degenerate. Therefore, $f(x) = \tr(\delta x^2)$ could be viewed as a non-degenerate quadratic form in $m$ indeterminates over $\gf_q$. Since $a$ is a nonsquare in $\gf_q^*$, we have $\tr(a \delta x^2) = b$ is equivalent to $\tr(\delta x^2) =  b a^{-1} $. Note that both $q$ and $m$ are odd. Then from Lemma~\ref{lem-quaform}, the conclusion follows.  
\end{proof}

Now we present the proof of the second assertion of Corollary~\ref{coro-sc1}.

 \begin{proof}[Proof of Corollary~\ref{coro-sc1}]
   By Theorem~\ref{thm-const1}, $N_0(\delta) = q^{m-1} - 1$ for each $\delta \in \gf_{q^m} \setminus \{0,1\}$ if $d_0 d_1$ is square. We now discuss the possible values of $N_b(\delta)$ for $b \in \gf_q^*$. 

If $\delta$ is a square, we have $\rho(\delta x) = \rho(x)$. Since $d_0d_1$ is a square, there are two cases. On one hand, if both $d_0$ and $d_1$ are squares in $\gf_{q^m}^*$, without loss of generality, suppose that $d_0 = u^2$ and $d_1 = v^2$ with $u, v \in \gf_{q^m}^*$, we then have 
\begin{eqnarray*}
  \lefteqn{ f(\delta x) - f(x) } \\
& = & \tr( (\delta - 1) \rho(x) x ) \\
  & = & \left\{ \begin{array}{ll}
    \tr( (\delta-1)d_0 y^2), & \textrm{ if $x = y^2$,} \\
    a \tr( (\delta-1)d_1 y^2), & \textrm{ if $x = a y^2$,}
  \end{array} \right. \\
  & = & \left\{ \begin{array}{ll}
    \tr( ( \delta - 1) u^2 y^2), & \textrm{ if $x = y^2$,} \\
    a \tr( ( \delta - 1) v^2 y^2), & \textrm{ if $x = a y^2$,}
  \end{array} \right. \\
  & = & \left\{ \begin{array}{ll}
    \tr( ( \delta - 1) (uy)^2), & \textrm{ if $x = y^2$,} \\
    a \tr( ( \delta - 1) (vy)^2), & \textrm{ if $x = a y^2$,}
  \end{array} \right.
\end{eqnarray*}
where $a \in \gf_q^*$ is a nonsquare. It then follows from Lemma~\ref{lem-pre3} that
\begin{equation*}
  N_b(\delta) = \frac{q^{m-1} + q^{{(m-1)}/2}}{2} +  \frac{q^{(m-1)} - q^{(m-1)/2}}{2} = q^{m-1}.
\end{equation*}
On the other hand, if $d_0$ and $d_1$ are both nonsquares, the argument is similar and we also obtain 
$$
N_b(\delta) = q^{m-1}.
$$

If $\delta$ is a nonsquare, we have
\begin{eqnarray} \label{eqn-const11}
  \lefteqn{f(\delta x) - f(x)} \nonumber \\
  & = & \tr(\delta x \rho(\delta x) - x \rho (x) ) \nonumber \\
  & = & \left\{ \begin{array}{ll}
    \tr( ( \delta d_1 - d_0) y^2), & \textrm{ if $x = y^2$,} \\
    a \tr( (\delta d_0 - d_1) y^2), & \textrm{ if $x = ay^2$,}
  \end{array} \right.
\end{eqnarray}
where $a \in \gf_q^*$ is a nonsquare. By Lemma~\ref{lem-pre3} and (\ref{eqn-const11}), we have $N_b(\delta) = q^{m-1}$ if and only if 
$$
\eta( \delta d_1 - d_0 ) = \eta( \delta d_0 - d_1 ),
$$
where $\eta$ is the quadratic character of $\gf_{q^m}$. This means that both of the following two systems of equations
\begin{equation}\label{eqn-const12}
  \left\{ \begin{array}{l}
    a z^2 d_1 - d_0 = x^2 \\
    a z^2 d_0 - d_1 = a y^2
  \end{array} \right.
\end{equation}
and
\begin{equation}\label{eqn-const13}
  \left\{ \begin{array}{l}
    a z^2 d_1 - d_0 = a x^2 \\
    a z^2 d_0 - d_1 = y^2 
  \end{array} \right.
\end{equation}
have no solution, where $a$ is a nonsquare in $\gf_q^*$. The system of equations (\ref{eqn-const12}) is equivalent to
\begin{equation*}
  \left\{ \begin{array}{l}
    x^2 = - d_0  + a d_1 z^2 \\
    y^2 = - d_1/a + d_0 z^2 . 
  \end{array} \right. 
\end{equation*}
Then by Lemma~\ref{lem-numsol}, the number $N_1$ of solutions of (\ref{eqn-const12}) satisfies 
$$
| N_1 - q^m | \le C q^{m/2},
$$
for some constant $C$ independent of $q$ when $d_0 \ne \pm d_1$. Thus, for a large enough $q^m$, we can always choose suitable $d_0$ and $d_1$ such that $N_1 \ne 0$. Then we have $N_b(\delta) \ne q^{m-1}$ for each $b \in \gf_q^*$, which completes the proof. 
\end{proof}

\begin{remark}\label{rmk-sc1}
  \begin{itemize}
    \item[a)] The trace function can be viewed as a subcase of the construction of ZDB functions in Corollary~\ref{coro-sc1} (if $d_0 = d_1$, also see~\cite{ZTWY12}). We note that this construction is new since for large $q^m$, we can always choose suitable $d_0$ and $d_1$ such that the ZDB functions are not difference balanced, while all previously known ZDB functions with the same parameters are difference balanced.
    \item[b)] Since every ZDB function $f(x)$ constructed in Corollary~\ref{coro-sc1} has the parameters $(q^m-1, q, q^{m-1}-1)$, by Lemma~\ref{lem-lowbound}, there are $q-1$ preimage sets of size $q^{m-1}$ and the rest one preimage set of size $q^{m-1} - 1$. 
  \end{itemize}
\end{remark}

\begin{example}\label{exm-sc1}
  Let $q=3$, $m=3$. Define $d_0 := 1$, $d_1 :=\theta^2$ where $\theta$ is a root of the irreducible polynomial $x^3 + 2x + 1 \in \gf_q[x]$. Then for the function $f: \gf_{q^m}^* \rightarrow \gf_q$, defined as in Corollary~\ref{coro-sc1}, $N_0(\delta) = 9$ for each $\delta \in \gf_{3^3} \setminus \{0,1\}$, and the distribution of $N_b(\delta)$ for all $b\neq 0$ is:
    \begin{table}[h]
    \begin{center}
    \begin{tabular}{|c|c|c|c|}
      \hline
      $N_b(\delta)$ & 6 & 9 & 12  \\\hline
      multiplicity & 4 & 17 & 4  \\
      \hline
    \end{tabular}
    \end{center}
    \end{table}
\end{example}

\subsubsection{Special case II}
Let $q$ be a prime power and $u = 1$. We have the second special case of Theorem~\ref{thm-const1} as follows. 

\begin{corollary}\label{coro-sc2}
  Let $q$ be a prime power, $e$ be a divisor of $q-1$ with $\gcd(e, m) = 1$ and $e = l\cdot r$. Let $D_0 = \langle \theta^r \rangle$, $C_0 = \langle \theta^e \rangle$, and $\alpha = \theta^{\frac{q^m-1}{q-1}}$. Define the function $f: D_0 \rightarrow \gf_q$ by
  $$
   f(x) :=  \tr(\rho(x) x), 
  $$
and $\rho(x)$ is defined as
$$
\rho(x) := d_i, \ \textrm{ if $x \in C_i$} ,
$$
where $C_i = \alpha^{ir} C_0$ and $d_i \in \gf_{q^m}^*$ for $0 \le i \le l - 1$. If $d_j / d_{k+j} \not \in C_k$ for each $k \ne 0$ and $0 \le j < l$, then the function $f(x)$ is a $\left(\frac{q^m - 1}{r}, q, \frac{q^{m-1} - 1}{r}\right)$-ZDB function.
\end{corollary}

\begin{proof}
  The conclusion follows from Theorem~\ref{thm-const1}.
\end{proof}

\begin{remark}\label{rmk-sc2}
  The construction in~\cite[Theorem 9]{Ding09} can be viewed as a subcase of the construction of ZDB functions given in Corollary~\ref{coro-sc2} (if $d_0 = d_1 = \cdots = d_{l-1}$, see also~\cite[Proposition 7]{Ding08}).  
\end{remark}

We give the following example to compare our construction in Corollary~\ref{coro-sc2} with the construction in~\cite[Theorem 9]{Ding09}. 

\begin{example}\label{exm-sc2}
  Let $q = 3^2$, $m = 3$, $l = r = 2$, $e = 4$, and $\theta$ be a root of the irreducible polynomial $x^6 + 2 x^4 + x^2 + 2x + 2 \in \gf_3[x]$. Define $\rho(x)$ as
  \begin{equation*}
    \rho(x) := \left\{ \begin{array}{ll}
      \theta^4, \textrm{ if $x \in \langle \theta^4 \rangle$,}\\
      \theta^8, \textrm{ if $x \in \theta^2 \langle \theta^4 \rangle$.}
    \end{array}\right.
  \end{equation*}
Then for the function $f : D_0 = \langle \theta^2 \rangle \rightarrow \gf_q$, defined in Corollary~\ref{coro-sc2}, $N_0(\delta) = 40$, and for $b \ne 0$, $N_b(\delta)$ has exactly three possible values: $36$, $45$, and $54$; in comparison, for the function $f : D_0 \rightarrow \gf_q$ defined in~\cite[Theorem 9]{Ding09}, $N_0(\delta) = 40$, and for $b \ne 0$, $N_b(\delta)$ has only two possible values: $36$ and $45$.
\end{example}





\section{New sets of ZDB functions}\label{sec-exd}

The construction of ZDB functions in Theorem~\ref{thm-const1} can generate many new single ZDB functions. In this section, we show that it can be extended in a natural way to construct a set of ZDB functions in which any two distinct ZDB functions are also related uniformly. Furthermore, we present some constructions of ZDB functions with flexible parameters.

\subsection{The construction}

\begin{theorem}\label{thm-const3}
   With the same notations as in Theorem~\ref{thm-const1}, define the set $\cals := \{f_i: 0 \leq i < r \}$, and each $f_i: (\bZ_n, +) \rightarrow (\gf_q,+)$ where $n = \frac{q^m-1}{r}$ as
  $$
  f_i(t) := \tr( \alpha^i \rho(t) \theta^{rut}),
  $$
  where $\rho(t)$ is defined as
  $$
  \rho(t) = d_i, \ \textrm{ if $\theta^{rt} \in C_i$,}
  $$
  with $C_i = \alpha^{i r} C_0$ and $d_i \in D_0$ for $0 \leq i < l$ . If the two following conditions
  \begin{enumerate}
    \item[(i)] $\{x \in C_0 :  x^u=1 \textrm{ and } x\neq 1\} = \emptyset$;
    \item[(ii)] $d_j / d_{k+j} \not \in C_{uk}$ for each $k \ne 0$ and $0 \le j < l$, where the subscripts $uk$ and $k+j$ are performed modulo $l$,
  \end{enumerate}
  are satisfied, then each function $f_i(t) \in \cals$ is a $\left(\frac{q^m - 1}{r}, q, \frac{q^{m-1} - 1}{r}\right)$-ZDB function, and any two distinct functions $f_{i_1}(t), f_{i_2}(t) \in \cals$ satisfy
$$
\big| \{ t \in \bZ_n: f_{i_1}( t + a ) - f_{i_2}(t) = 0 \} \big| = \frac{q^{m-1} - 1}{r},
$$
for $0 \leq i_1 \ne i_2 < r$ and every $a \in \bZ_n$. 
  \end{theorem}
 
\begin{proof}
  By definition, $f_i(t) = \alpha^i \tr( \rho(t) \theta^{rut})$. Then from Theorem~\ref{thm-const1} it follows that each $f_i(t) \in \cals$ is a $\left( \frac{q^m - 1}{r}, q, \frac{q^{m-1} - 1}{r} \right)$-ZDB function if the conditions (i) and (ii) are satisfied.

  For any two distinct functions $f_{i_1}(t), f_{i_2}(t) \in \cals$, without loss of generality, assume that $\theta^{ra} \in C_k$ for some $0 \leq k < l$. We then have
  \begin{eqnarray*}
    \lefteqn{ \big| \{ t \in \bZ_n: f_{i_1}(t + a) - f_{i_2}(t) = 0 \}\big|}  \\
    & = & \big| \{ t \in \bZ_n : \alpha^{i_1} \tr\left( (\rho(t+a) \theta^{rau} - \rho(t) \alpha^{ {i_2} - {i_1} } ) \theta^{rut} \right) = 0 \} \big| \\
    & = & \sum_{j = 0}^{l- 1} \big| \{ x \in C_j : \tr\left( ( d_{k+j} \theta^{rau} - d_j \alpha^{i_2 - i_1} ) x^u \right) = 0 \} \big| . 
  \end{eqnarray*}
  If $k = 0$, i.e., $\theta^{ra} \in C_0$, suppose that $d_j \theta^{rau} - d_j \alpha^{i_2 - i_1} = 0$ for some $0 \leq i_1 \ne i_2 < r$ and $\alpha = \theta^{\frac{q^m-1}{q-1}}$, which means there exists some $0 \leq c < \frac{q^m-1}{e}$, such that 
  \begin{equation}\label{eqn-setprf1}
  c \cdot e \cdot u \equiv \frac{q^m-1}{q-1} \cdot i \pmod{q^m-1},
  \end{equation}
for some $i = \pm 1, \pm 2, \ldots, \pm (r-1)$. Since $\gcd(e,m) = \gcd(u,m) = 1$, both $e$ and $u$ are co-prime to $\frac{q^m - 1}{q-1}$. Thus, $c$ in (\ref{eqn-setprf1}) must possess a divisor $\frac{q^m - 1}{q-1}$. The relation (\ref{eqn-setprf1}) is then equivalent to that there exists a $0 \leq c' < \frac{q-1}{e}$, such that
  \begin{equation}\label{eqn-setprf2}
    e \cdot u \cdot c' - i  \equiv 0 \pmod{q-1},
  \end{equation}
  for some $i = \pm 1, \pm 2, \ldots, \pm (r-1)$. However, since $e \nmid i$, (\ref{eqn-setprf2}) cannot hold anyway. Therefore, $d_j \theta^{rau} - d_j \alpha^{i_2 - i_1} \ne 0$ for $\theta^{ra} \in C_0$ and any $0 \leq i_1 \ne i_2 < r$. Then by Lemma~\ref{lem-pre2}, we have 
  $$
  \big| \{ t \in \bZ_n: f_{i_1}(t+a) - f_{i_2}(t) = 0 \} \big| = \frac{q^{m-1}-1}{r},
  $$
  for $\theta^{ra} \in C_0$ and any $0 \leq i_1 \ne i_2 < r$.

  If $k \ne 0$, since $d_i \in D_0$ for each $0 \leq i < l$, by Lemma~\ref{lem-pre}, we have $d_{k+j} \theta^{rau} - d_j \alpha^{i_2 - i_1} \ne 0$ for each $\theta^{ra} \in C_k$ and any $0 \leq i_1 \ne i_2 < r$. By Lemma~\ref{lem-pre2}, we also have
  $$
  \big| \{ t \in \bZ_n: f_{i_1}( t+a ) - f_{i_2}(t) = 0 \} \big| = \frac{q^{m-1}-1}{r},
  $$
  for $\theta^{ra} \in C_k$ with $0 < k < l$ and any $0 \leq i_1 \ne i_2 < r$.
    The proof is then completed.
\end{proof}

\begin{remark}\label{rmk-const3}
  According to Remark~\ref{rmk-const1}, the two sufficient conditions in Theorem~\ref{thm-const3} can be satisfied easily, and there are exponentially many $\rho(t)$'s satisfying the conditions. 
\end{remark}

The following construction of sets of ZDB functions is more general.

\begin{corollary}\label{coro-const3}
  Let $\{g_0, g_1, \ldots, g_{r-1}\}$ be a complete set of representatives for the cyclotomic classes of order $r$ in $\gf_{q^m}$. Define the set $\cals := \{f_i: 0 \leq i < r \}$, and each $f_i: (\bZ_n,+) \rightarrow (\gf_q,+)$ where $n = \frac{q^m-1}{r}$ as
  $$
  f_i(t) := \tr( g_i \rho(t) \theta^{rut}),
  $$
  where $\rho(t)$ is defined as
  $$
  \rho(t) = d_i, \ \textrm{ if $\theta^{rt} \in C_i$,}
  $$
  with $C_i = \alpha^{i r} C_0$, $\alpha = \theta^{\frac{q^m-1}{q-1}}$, and $d_i \in D_0$ for $0 \leq i < l$ . If the following two conditions
  \begin{enumerate}
    \item[(i)] $\{x \in C_0 : x^u=1 \textrm{ and } x\neq 1\} = \emptyset$;
    \item[(ii)] $d_j / d_{k+j} \not \in C_{uk}$ for each $k \ne 0$ and $0 \le j < l$, where the subscripts $uk$ and $k+j$ are performed modulo $l$,
  \end{enumerate}
  are satisfied, then each function $f_i(t) \in \cals$ is a $\left(\frac{q^m - 1}{r}, q, \frac{q^{m-1} - 1}{r}\right)$-ZDB function, and any two distinct functions $f_{i_1}(t), f_{i_2}(t) \in \cals$ satisfy
$$
\big| \{ t \in \bZ_n : f_{i_1}( t + a ) - f_{i_2} (t) = 0 \} \big| = \frac{q^{m-1} - 1}{r},
$$
for each $0 \leq i_1 \ne i_2 < r$ and every $a \in \bZ_n$. 
  \end{corollary}

  \begin{proof}
    Without loss of generality, suppose that $g_i \in D_i$. By Lemma~\ref{lem-pre}, we have $g_i = \alpha^i g_i'$ where $g_i' \in D_0$. The proof is then straightforward from that of Theorem~\ref{thm-const3}.  
  \end{proof}

  \begin{remark}\label{rmk-coro}
    The construction in Corollary~\ref{coro-const3} can be viewed as a generalization of the existing constructions in~\cite{DMY07,DY08,GMY09} (if $d_0 = d_1 = \cdots = d_{l-1}$). Furthermore, Theorem~\ref{thm-lc} in Section~\ref{sec-app} indicates that the construction in Theorem~\ref{thm-const3} can really generate many new classes of sets of ZDB functions.
  \end{remark}

To illustrate the generic construction in Corollary~\ref{coro-const3}, we give the following example. 

\begin{example}\label{exm-zdbset1}
  Let $q = 3^2$, $m = 3$, $l = r = 2$, $e = 4$, $u = 1$, and $\theta$ be a root of the irreducible polynomial $x^6 + 2 x^4 + x^2 + 2x + 2 \in \gf_3[x]$. Define $\rho(t)$ as 
\begin{equation*}
  \rho(t) := \left\{ \begin{array}{ll}
    \theta^4, & \textrm{ if $rt \equiv 0 \pmod{e}$,}\\
    \theta^8, & \textrm{ if $rt \equiv r \pmod{e}$.}
  \end{array}\right.
\end{equation*}
Then the set of ZDB functions is defined as
$$
\cals := \{ f_0, \ f_1 \},
$$
where $f_0(t) := \tr\left( \rho(t) \theta^{rt}\right)$, and $f_1(t) := \tr\left( \theta^{91} \rho(t) \theta^{rt}\right)$. The $f_i(t)$ is a $(364, 9, 40)$-ZDB function for $i = 1, 2$, and 
$$
\big| \{ t \in \bZ_{364} : f_0(t+a) - f_1 (t) = 0 \} \big| = 40,
$$
for each $a \in \bZ_{364}$.
\end{example}

\subsection{ZDB functions with flexible parameters}


In~\cite{ZTWY12}, difference balanced functions were used to construct ZDB functions with flexible parameters. It turns out that the functions given in Theorem~\ref{thm-const1} could also be employed to construct ZDB functions with parameters $\left( \frac{q^m-1}{r}, q^v, \frac{q^{m-v} - 1}{r} \right)$, and further can generate a set of ZDB functions with such parameters.  

\begin{theorem}\label{thm-const2}
  With the same notations as in Theorem~\ref{thm-const1}, suppose that $f(t) = \tr(\rho(t) \theta^{rut})$ is a $\left( \frac{q^m - 1}{r}, q, \frac{q^{m-1}-1}{r} \right)$-ZDB function from $(\bZ_n,+)$ onto $(\gf_q,+)$ defined in Theorem~\ref{thm-const1}, where $n = \frac{q^m - 1}{r}$. Let $a_0, a_1, \ldots, a_{v-1}$ be $v$ elements in $\gf_{q^m}^*$, which are linearly independent over $\gf_q$. Define the function $f_v: (\bZ_n, +) \rightarrow (\gf_q,+)^v$ as
  $$
  f_v(t) := \left( \tr(a_0 \rho(t) \theta^{rut}), \tr(a_1 \rho(t)\theta^{rut}), \ldots, \tr(a_{v-1} \rho(t) \theta^{rut}) \right),
  $$
  then the function $f_v(t)$ is a ZDB function with parameters $\left( \frac{q^m-1}{r}, q^v, \frac{q^{m-v} -1}{r} \right)$. 
\end{theorem}

Similar to the proof of Theorem~\ref{thm-const1}, using the result on the number of solutions of linear systems, one can easily give a proof for Theorem~\ref{thm-const2}.


\begin{corollary}\label{coro-vec}
  Suppose that $\cals = \{f_0, f_1, \ldots, f_{r-1} \}$ is the set of ZDB functions constructed in Corollary~\ref{coro-const3}, i.e., $f_i (t) = \tr(g_i \rho(t) \theta^{rut})$, where $\{g_0, g_1, \ldots, g_{r-1} \}$ is a complete set of representatives for the cyclotomic classes of order $r$ in $\gf_{q^m}$. Let $a_0, a_1, \ldots, a_{v-1}$ be $v$ elements in $\gf_{q^m}^*$, which are linearly independent over $\gf_q$. Define the set $\cals '$ of ZDB functions as $\cals ' := \{  f_0', f_1', \ldots, f_{r-1}' \}$, where $f_i': (\bZ_n, +) \rightarrow (\gf_q,+)^v$ is
  $$
  f_i'(t) := \left(\tr( a_0 g_i \rho(t) \theta^{rut}), \tr( a_1 g_i \rho(t) \theta^{rut} ), \ldots, \tr(a_{v-1} g_i \rho(t) \theta^{rut} ) \right) .
  $$
  Then the set $\cals '$ is a set of r ZDB functions with parameters $\left( \frac{q^m-1}{r}, q^v, \frac{q^{m-v}-1}{r} \right)$, and any two distinct functions $f_{i_1}'(t), f_{i_2}'(t) \in \cals '$ satisfy
$$
\big| \{ t \in \bZ_n : f_{i_1}'( t + a ) - f_{i_2} '(t) = 0 \} \big| = \frac{q^{m-v} - 1}{r},
$$
for $0 \leq i_1 \ne i_2 < r$ and every $a \in \bZ_n$.
\end{corollary}

With a set of ZDB functions, using the idea in~\cite[Theorem 6]{ZTWY12}, we can give a new construction of ZDB functions with more flexible parameters. 

\begin{theorem}\label{thm-const4}
  Suppose that $f_0', f_1', \ldots, f_{k-1}'$ are any $k$ functions in the set of ZDB functions constructed in Corollary~\ref{coro-vec} with $1 \leq k \leq r$ and $\gcd(k, n) = 1$ where $n = \frac{q^m-1}{r}$. Define the function $f: (\bZ_{kn},+) \rightarrow (\gf_q^v,+)$ as $f(t) := f_i'(j)$, where $t = j k + i$ with $j \in \bZ_n$ and $i \in \bZ_k$. Then $f(t)$ is a $\left( k \frac{q^m-1}{r}, q^v, k \frac{q^{m-v} -1 }{r} \right)$-ZDB function.
\end{theorem}

\begin{proof}
  For each nonzero $a \in \bZ_{kn}$, since $\gcd(k,n) = 1$, we may write $a = a_1k + a_2$ where $(a_1, a_2) \in \bZ_n \times \bZ_k$ and $a_1 \ne 0$ or $a_2 \ne 0$. Note that
  \begin{eqnarray*}
    \lefteqn{\big| \{  t \in \bZ_{kn}: f(t + a) - f(t) = 0 \} \big| } \\
    & = & \big| \{ (j, i) \in \bZ_n \times \bZ_k : f(jk + i + a_1k + a_2) - f(jk + i) = 0 \} \big| .
  \end{eqnarray*}
If $a_2 = 0$ and $a_1 \ne 0$, we have
\begin{eqnarray*}
  \lefteqn{\big| \{ t \in \bZ_{kn} : f(t+a) - f(t) = 0 \} \big| } \\
  & =& \sum_{i=0}^{k-1} \big| \{ j \in \bZ_n : f_i'(j+a_1) - f_i'(j) = 0 \} \big| \\
  & = & k \frac{q^{m-v} - 1}{r} .
\end{eqnarray*}
If $a_2 \ne 0$, we have
\begin{eqnarray*}
  \lefteqn{\big| \{ t \in \bZ_{kn} : f(t+a) - f(t) = 0 \} \big| } \\
  & =& \sum_{i=0}^{k-1-a_2} \big| \{ j \in \bZ_n : f_{i+a_2}'(j+a_1) - f_i'(j) = 0 \} \big| \\ 
  & & \ + \sum_{i=k - a_2}^{k-1} \big| \{ j \in \bZ_n : f_{i+a_2-k}'(j+a_1+1) - f_i'(j) = 0 \} \big| \\
  & = & k \frac{q^{m-v} - 1}{r} .
\end{eqnarray*}
The proof is then completed.
\end{proof}

\section{Two applications of sets of ZDB functions}\label{sec-app}

In this section, we present two applications of sets of ZDB functions: one is optimal sets of frequency hopping (FH) sequences, and the other is optimal constant weight codes. In the literature, ZDB functions or corresponding PDFs have been used to construct optimal frequency-hopping sequences~\cite{DMY07,DY08,FMM04,GFM06,GMY09}.

\subsection{Optimal sets of frequency hopping sequences}

In frequency hopping (FH) CDMA communication systems, a transmitter changes its carrier frequency at regular intervals as prescribed by an FH sequence~\cite{SOSL02}. Let $B = \{b_0, b_1, \ldots, b_{\ell - 1}\}$ be a set of available frequencies (also called {\em alphabet}) and $(s_0, s_1, \ldots, s_{n-1})$ be an FH sequence of length $n$ over $B$, where $s_i \in B$. In FH CDMA communication systems, long messages are transmitted by repeating the FH sequence as often as necessary. For any two FH sequences $X, Y$ of length $n$ over $B$, their Hamming correlation $H_{X, Y}$ is defined as
$$
H_{X,Y}(t) := \sum^{n-1}_{i=0} h[x_i, y_{i+t}], \quad 0 \leq t < n
$$
where $h[a,b] = 1$ if $a = b$, and $0$ otherwise, and all operations among the position indices are performed modulo $n$. To maximize the throughput, the Hamming correlation is required as small as possible. For one single FH sequence, in 1974, Lempel and Greenberger developed the following lower bound~\cite{LG74}. 

\begin{lemma}\label{lem-fhsbound}
For every FH sequence $X$ of length $n$ over an alphabet of size $\ell$, define
\begin{equation*}
H(X) := \max_{1 \leq t < n} \{H_{X,X} (t) \},
\end{equation*}
then
\begin{equation}\label{eqn-fhsbound}
H(X) \geq \left\lceil \frac{(n - \epsilon) ( n + \epsilon - \ell)}{\ell (n - 1)} \right\rceil ,
\end{equation}
where $\epsilon$ is the least nonnegative residue of $n$ modulo $\ell$.
\end{lemma}

Let $(n, \ell, \lambda)$ denote an FH sequence $X$ of length $n$ over an alphabet of size $\ell$ with $\lambda = H(X)$. In Section~\ref{sec-char}, the lower bound on $\lambda$ of ZDB functions in Lemma~\ref{lem-lowbound}, in fact coincides with the lower bound of (\ref{eqn-fhsbound}). A set $\calf$ of FH sequences is call {\em optimal}, if one of the following bounds on $M(\calf)$ is met, where
$$
M(\calf) := \max \left\{ \max_{X \in \calf} H(X), \max_{X,Y \in \calf, X \ne Y} H(X,Y) \right\} ,
$$
and $H(X,Y) := \max_{0 \leq t < n} \{ H_{X,Y} (t) \}$. By convention, let $(n, N, \lambda; \ell)$ denote a set of $N$ FH sequences of length $n$ over an alphabet of size $\ell$, where $\lambda = M(\calf)$.

\begin{lemma}\cite{PF04,Sar05}\label{lem-fhssetbound}
  Let $\calf$ be a set of $N$ sequences of length $n$ over an alphabet size of $\ell$. Define $I := \lfloor nN / \ell \rfloor$. Then
  $$
  M(\calf) \geq \left\lceil \frac{(nN - \ell) n}{(nN - 1) \ell} \right\rceil 
  $$
  and
  $$
  M(\calf) \geq \left\lceil \frac{2I n N - (I+1) I \ell}{(nN - 1)N} \right\rceil .
  $$
\end{lemma}

By the definition of sets of ZDB functions, we have the following bridge between sets of ZDB functions and sets of FH sequences.

\begin{lemma}\label{lem-zdbfhsset}
  Suppose that $\cals = \{f_0, f_1, \ldots, f_{N-1} \}$ is a set of $N$ $(n, \ell, \lambda)$-ZDB functions from $(\bZ_n,+)$ onto an abelian group $(B, +)$ of order $\ell$. Define the sequence set $\calf := \{ {\bf s}_0, {\bf s}_1, \ldots, {\bf s}_{N-1} \} $, where $s_i(t) := f_i(t)$ for $0 \leq i < N$ and $0 \leq t < n$. Then $\calf$ is an $(n, N, \lambda; \ell)$ set of FH sequences. 
\end{lemma}

Using our construction of sets of ZDB functions, we can construct optimal sets of FH sequences, of which each FH sequence is also optimal with respect to the bound of (\ref{eqn-fhsbound}).

\begin{theorem}\label{thm-fhsset}
  Suppose that $\cals = \{f_0, f_1, \ldots, f_{r-1}\}$ is the set of ZDB functions constructed in Corollary~\ref{coro-vec}. Define the set of sequences 
  $$
  \calf := \{ {\bf s}_0, {\bf s}_1, \ldots, {\bf s}_{r-1} \},
  $$
  where $s_i(t) := f_i(t)$ for $0 \leq i < r$ and $0 \leq t < \frac{q^m-1}{r}$. Then $\calf$ is an optimal set of FH sequences with parameters $\left(\frac{q^m-1}{r}, r, \frac{q^{m-v} -1}{r}; q^v \right)$. Furthermore, each ${\bf s}_i$ for $0 \leq i <r$ is
  an optimal $\left( \frac{q^m-1}{r}, q^v, \frac{q^{m-v}-1}{r} \right)$ FH sequence. 
\end{theorem}

In applications, FH sequences over a finite field are required to have large linear complexity~\cite{Kumar88}. For a sequence ${\bf s} = (s_t)$ of period $N$ over a finite field ${\mathbb F}$, the {\em linear complexity} $\lc({\bf s})$ is defined to be the least positive integer $L$ such that there exist constants $c_0 = 1$, $c_1, \ldots, c_L \in {\mathbb F}$ such that
$$
- s_i = c_1 s_{i-1} + c_2 s_{i-2} + \cdots + c_L s_{i-L}
$$
for all $i \ge L$. A polynomial of the form
$$
M(x) = c_0 + c_1 x + \cdots + c_L x^L \in {\mathbb F}[x],
$$
is called the {\em minimal polynomial} of the sequence ${\bf s}$. The following lemma is useful to determine the minimal polynomial and the linear complexity.

\begin{lemma}\label{lem-lc}\cite{AB92}
  Every sequence ${\bf s} = (s_t)$ over $\gf_q$ of period $q^m - 1$ has a unique expansion of the form 
  $$
  s_t = \sum^{q^m - 2}_{i=0} c_i \beta^{it}, \textrm{ for all $0 \leq t \leq q^m - 2$},
  $$
  where $\beta$ is a primitive element of the extension field $\gf_{q^m}$ and $c_i \in \gf_{q^m}$ for $0 \leq i \leq q^m - 2$. Define the index set $I := \{i : c_i \ne 0, \ 0 \leq i \leq q^m - 2 \}$, then the minimal polynomial $M(x)$ of the sequence ${\bf s}$ is
  $$
  M(x) = \prod_{i \in I} (x - \beta^i ) , 
  $$
  and the linear complexity of ${\bf s}$ is the cardinality $|I|$ of the set $I$.
\end{lemma}

To determine the linear complexity of the FH sequences generated by Theorem~\ref{thm-fhsset}, we also need the following lemma.

\begin{lemma}\label{lem-cycmappoly}\cite{NW05}
  For a positive divisor $e$ of $q-1$ and $d_0, d_1, \ldots, d_{e-1} \in \gf_q$, the cyclotomic mapping polynomial $f_{d_0, d_1, \ldots, d_{e-1}} = \rho(x) x^u$ is given by
$$
f_{d_0, d_1, \ldots, d_{e-1}} = (a_{e-1} x^{(e-1)(q-1)/e} + \cdots + a_1 x^{(q-1)/e} + a_0 ) x^u
$$
with 
$$
a_i = e^{-1} \sum_{j=0}^{e-1} d_j \alpha^{-ij(q-1)/e}, \textrm{ $i = 0, 1, \ldots, e - 1$},
$$
where $e^{-1}$ denotes the inverse of $e$ modulo the characteristic of $\gf_q$, and $\alpha$ is a primitive element of $\gf_q$.
\end{lemma}
 
Now we are able to determine the linear complexity of the FH sequences in Theorem~\ref{thm-fhsset}.

\begin{theorem}\label{thm-lc}
  Let $\calf = \{ {\bf s}_0, {\bf s}_1, \ldots, {\bf s}_{r-1} \}$ be the set of FH sequences constructed in Theorem~\ref{thm-fhsset} with $v = 1$. Then the linear complexity of each sequence ${\bf s}_i \in \calf$ satisfies 
  $$
  m \leq \lc({\bf s}_i) \leq lm,
  $$
and both of the two equalities can be achieved by choosing suitable $\rho(t)$.
\end{theorem}

\begin{proof}
  By definition, ${\bf s_i} \in \calf$ is defined as
  $$
  s_i(t) := \tr(\alpha^i \rho(t) \theta^{rut}) ,
  $$
  where $\alpha = \theta^{\frac{q^m-1}{q-1}}$. By Lemma~\ref{lem-cycmappoly}, the cyclotomic mapping polynomial can be written as
  $$
  \rho(t) = a_{l-1} \theta^{(l-1)(q^m-1)t/l} + \cdots + a_1 \theta^{(q^m-1)t/l} + a_0 
  $$
  with 
  $$
  a_i = l^{-1} \sum_{j=0}^{l-1} d_j \theta^{-ij (q^m - 1)/l} ,
  $$
  where $l^{-1}$ denotes the inverse of $l$ modulo the characteristic of $\gf_q$, and $\theta$ is a primitive element of $\gf_{q^m}$. Thus, the sequence ${\bf s}_i$ can be written as
  \begin{eqnarray}\label{eqn-lc1}
    s_i(t) & = & \alpha^i \tr\left( \rho(t) \theta^{rut} \right) \nonumber \\
    & = & \alpha^i \tr \left( \sum_{j=0}^{l-1} a_j \theta^{(q^m-1)jt/l} \theta^{rut} \right) \nonumber \\
    & = & \alpha^i \sum_{k=0}^{m-1} \sum_{j=0}^{l-1} a_j^{q^k} \theta^{q^k ( j(q^m-1)/l + ru) t} .
  \end{eqnarray}
Suppose that there exist $0 \leq j_1, j_2 \leq l - 1$ and $0 \leq k_1, k_2 \leq m-1$, such that
$$
q^{k_1} (j_1 (q^m-1)/l + ru) \equiv q^{k_2} (j_2 (q^m-1)/l + ru ) \pmod{q^m - 1}.
$$
We then have 
\begin{equation}\label{eqn-lc2}
\frac{q^m-1}{l} q^{k_2} (q^{k_1 - k_2} j_1 - j_2) + ru q^{k_2} (q^{k_1 - k_2} - 1) \equiv 0 \pmod{q^m - 1}.
\end{equation}
It follows that 
$$
\frac{q^m-1}{l} \big|  ru q^{k_2} (q^{k_1-k_2} - 1),
$$
which holds if and only if $k_1 = k_2$ since $\gcd(e,m) = \gcd(u,m) = 1$ and $e = l\cdot r$. Back to (\ref{eqn-lc2}), we obtain $j_1 = j_2$. Hence, all the exponents of $\theta$ in (\ref{eqn-lc1}) are pairwise distinct. Then by Lemma~\ref{lem-lc}, we have
$$
\lc({\bf s}_i) = m \cdot |I|,
$$
where $I = \{ a_i \ne 0: \ 0 \leq i < l\}$ and $|I| \leq l$. Recall that
$$
a_i = l^{-1} \sum_{j=0}^{l-1} d_j \theta^{-ij (q^m - 1)/l} .
$$
It is easily seen that $|I| = 1$ if $d_0 = d_1 = \cdots = d_{l-1}$. We now argue that $a_i \ne 0$ for each $0 \leq i < l$ by choosing suitable $\rho(t)$ and $u$. Specifically, let $u = 1$ and $d_j = \theta^{rj}$ for $0 \leq j < l$. It is then checked that the two conditions in Theorem~\ref{thm-const3} are satisfied, and $a_i \ne 0$ for each $0 \leq i < l$. With such $\rho(t)$ and $u$, we have $\lc({\bf s}_i) = lm$ for each $0 \leq i < r$. The proof is then completed. 
\end{proof}

\begin{remark}\label{rmk-lc}
If $v = 1$, the construction in Theorem~\ref{thm-fhsset} generates optimal sets of FH sequences with the same parameters as~\cite[Theorem 4.7]{GMY09} (see also~\cite{DY08,DMY07}). In~\cite{Wang104}, it was determined that the linear complexity of FH sequences generated by \cite[Theorem 4.7]{GMY09} is $m$. Then by comparing the linear complexity of the generated FH sequences, Theorem~\ref{thm-lc} indicates that Theorem~\ref{thm-fhsset} can generate new optimal sets of FH sequences when $|I| > 1$.
\end{remark}

\subsection{Optimal constant weight codes}

An $(n,N,d,w)_\ell$ constant weight code is a code over an abelian group $\{b_0, b_1, \ldots, b_{\ell-1}\}$ with length $n$, size $N$, and minimum distance $d$ such that the Hamming weight of each codeword is the constant $w$. Let $A_\ell(n,d,w)$ denote the maximum size of an $(n, M, d, w)_\ell$ constant weight code. An $(n, M, d, w)_\ell$ constant weight code is called {\em optimal} if the following bound is met.

\begin{lemma}\cite{FVS98}\label{lem-cwcbound}
  If $nd - 2nw + \frac{\ell}{\ell - 1} w^2 > 0$, then 
  $$
  A_\ell (n,d,w) \leq \frac{nd}{nd - 2nw + \frac{\ell}{\ell - 1} w^2} .
  $$
\end{lemma}

Recently, Zhou et al. presented a method to construct constant weight codes from a set of ZDB functions~\cite{ZTWY12}. Using this method, we give the following construction of optimal constant weight codes.

\begin{theorem}\label{thm-cwc}
  Let $\cals$ be the set of ZDB functions constructed in Corollary~\ref{coro-vec}. For each $f_i \in \cals$ with $0 \leq i < r$, define a code $\calc_i$ as
  $$
  \calc_i := \left\{ c_j^i = (f_i(t_0 + t_j), \ldots, f_i(t_{n-1} + t_j)) :  t_j \in \bZ_n \right\}.
  $$
  Then the code $\calc := \bigcup_{i=0}^{r-1} \calc_i$ is an optimal constant weight code over $\gf_q^v$ with parameters 
  $$
  \left( \frac{q^m-1}{r}, q^m - 1, \frac{q^m - q^{m-v}}{r}, \frac{q^m- q^{m-v}}{r} \right)_{q^v} .  
  $$
\end{theorem}

\section{Concluding remarks}\label{sec-con}

In this paper, we summarized two results to characterize zero-difference balanced (ZDB) functions. As the main contribution, we presented a generic construction of single ZDB functions. Based on this construction, we further gave a generic construction of sets of ZDB functions. We also extended these two results to construct new ZDB functions with flexible parameters. As applications of sets of ZDB functions, we constructed optimal sets of FH sequences, and also optimal constant weight codes. Furthermore, by determining the linear complexity, we argued that our construct can generate many new optimal sets of FH sequences.

For the ZDB functions constructed in Theorem~\ref{thm-const1}, it seems hard to determine the sizes of the preimage sets explicitly. The sizes of the preimage sets are also important parameters, e.g., they constitute the parameter $K$ in the corresponding partitioned difference family. It would also be nice if the linear complexity of FH sequences generated by Theorem~\ref{thm-fhsset} could be determined explicitly.

\section*{Acknowledgments}

The authors are very grateful to the reviewers for their helpful and constructive comments.



\medskip
Received xxxx 20xx; revised xxxx 20xx.
\medskip

\end{document}